\documentclass[a4paper,12pt]{article}

\usepackage{graphicx}
\usepackage{amssymb}
\usepackage{amsmath}
\usepackage{verbatim}
\usepackage{algorithm}
\usepackage{algorithmic}

\newtheorem{definition}{Definition}

\newtheorem{lemma}{Lemma}

\newenvironment{proof}[1][Proof]{\begin{trivlist}
\item[\hskip \labelsep {\bfseries #1}]}{\end{trivlist}}

\newcommand{\C}{\mathcal{C}}
\newcommand{\ubalg}{\textsc{ub}~}
\newcommand{\mproalg}{\textsc{mp}~}
\newcommand{\probabalg}{\textsc{pr}~}
\newcommand{\Var}{\operatorname{Var}}
\newcommand{\E}{\operatorname{E}}
\newcommand{\Bin}{\operatorname{Bin}}
\renewcommand{\Pr}{\operatorname{Pr}}
\newcommand{\topk}{T}
\newcommand{\atopk}{\tilde{\topk}}

\newcommand{\mbf}{\mathbf}
\newcommand{\X}{\mbf{X}}

\newcommand{\w}{\mbf{w}}
\newcommand{\x}{\mbf{x}}

\newcommand{\acc}{\operatorname{acc}}
\newcommand{\cost}{\operatorname{cost}}
\newcommand{\dist}{\operatorname{dist}}

\begin{document}

\title{Approximate Top-$k$ Retrieval from Hidden Relations}
\author{
  Antti Ukkonen\\ \\
  \small{Aalto University and}\\\small{Helsinki Institute for
    Information Technology}
   \thanks{This work was conducted while the author was
  visiting Yahoo! Research Barcelona. The work was supported by
  Academy of Finland grant number 131276.}
   \\\small{\texttt{antti.ukkonen@hiit.fi}}
  }
\date{}

\maketitle

\begin{abstract}
We consider the evaluation of approximate top-k queries from relations
with a-priori unknown values.  Such relations can arise for example in
the context of expensive predicates, or cloud-based data sources.  The
task is to find an approximate top-$k$ set that is close to the exact
one while keeping the total processing cost low.  The cost of a query
is the sum of the costs of the entries that are read from the hidden
relation.

A novel aspect of this work is that we consider prior information
about the values in the hidden matrix.  We propose an algorithm that
uses regression models at query time to assess whether a row of the
matrix can enter the top-k set given that only a subset of its values
are known.  The regression models are trained with existing data that
follows the same distribution as the relation subjected to the query.

To evaluate the algorithm and to compare it with a method proposed
previously in literature, we conduct experiments using data from a
context sensitive Wikipedia search engine. The results indicate that
the proposed method outperforms the baseline algorithms in terms of
the cost while maintaining a high accuracy of the returned results.

\end{abstract}

\section{Introduction}

Databases are
traditionally concerned with the task of
efficiently retrieving a set of tuples that
match a given query.
However, we can also
rank the result set 
according to some scoring function.
Often especially the top-scoring tuples
of this ranking
are of interest.
A typical example of this are information retrieval systems,
as the user of a search engine is unlikely to be interested in
the entire ranking of results,
but the first few items only.
The idea of top-$k$ processing
(see e.g.~\cite{fagin2001,hristides01sigmod,natsev01vldb,theobald04topk,marian04tods,das06answering,kumar09topk})
is to retrieve the $k$ first results of this ranking
without computing the score for every matching document.
Most of existing work on top-$k$ search
focuses on the task of retrieving
tuples with the highest score according to
a given scoring function $f$
from a known relation.
In general $f$ is
assumed to be monotonic \cite{fagin2001},
and usually it is a
convex combination of the attribute values.
Moreover, since the relation is known,
various indexing techniques can be applied
to speed up the processing.
For a more in-depth discussion,
see \cite{topksurvey} for
an excellent survey on the topic.

In this paper we consider the top-$k$ problem
in a somewhat different scenario
that can be explained as follows.
Suppose we have two people, $A$ and $B$.
The person $A$ has a
vector of $m$ elements,
while the person $B$ has a
matrix of $n$ rows and $m$ columns.
The task of person $A$ is to find
those $k$ rows of this matrix
that have the highest
inner product with her vector.
However, $B$ will not reveal the matrix to $A$.
Instead, $A$ has to ask for the values of
individual cells one by one.
Moreover, to each cell is associated a cost that
$A$ has to pay before $B$ reveals the value.
These costs are known to $A$.
How should $A$ proceed in order to
find the $k$ highest scoring rows
while keeping the total cost of the process low?
We also give $A$ some additional information
in the form of examples of matrices
that $B$ has had in the past.
Using these as training data $A$ can
employ machine learning techniques to
find out what elements of the matrix to ask for.

In other words,
we must apply a known linear ranking function
(the vector of person $A$)
to a hidden relation
(the matrix of person $B$)
given knowledge
about the distribution
of the values in
the hidden relation.
This is in contrast to
most of existing work
where the relation is assumed to be known,
and the ranking function may vary.
We also assume that
the access costs are high:
reading the value of an entry in the matrix
is computationally expensive.
In this paper we
propose an algorithm
that will find an approximate answer
to a top-$k$ query
while keeping the cost of the query low.

Since the contents
of the hidden relation
are unknown at the time the query is issued,
a solution can not rely on
pre-built index structures.
We do assume, however,
that all relations that we will encounter
follow the same distribution,
and that we can sample
training data from this distribution.
The algorithm that we propose
makes use of regression models
to estimate whether or not
a row can belong to the top-$k$ set
after having observed only
a subset of its entries.
Moreover,
we can decrease the cost of the query
by allowing a small number of errors
in the results.
That is, we allow the algorithm to
return a set of documents that is not the
exact top-$k$ set.
The results may miss some high scoring documents,
and respectively contain other documents that
do not belong to the exact top-$k$ set.

In practice
this top-$k$ problem can be motivated
by the following design of a
context sensitive search engine \cite{ukkonen08cikm},
where a query is assumed to consist of
a set of terms and a \textit{context document}.
The context document can be
e.g.~the page currently viewed by the user.
To process the query we first
retrieve all documents
that contain the query terms,
and then describe each of these by a feature vector
that is a function of the context.
The final score of a document is given
by the inner product of the feature vector
with a scoring vector.
A context dependant feature
could be e.g.~a measure of textual similarity
between the context document
and the document that is being ranked.
These context dependant features could be implemented as expensive predicates.
We can thus think that
person $B$ is hiding the feature vectors,
and by computing a feature we are asking for its value.
The total cost we have to pay
reflects the computational overhead
associated with finding out the values of the features.


\subsection{Related Work}


We discuss related work from a number of different angles:
databases, 
approximate nearest neighbor search, and
learning theory.

\subsubsection{Databases}

A considerable amount of literature has been published about
top-$k$ processing in the past years.
For a thorough review we refer the reader
to the survey by Ilyas et.~al.~\cite{topksurvey}.
Usually the basic setting in these is slightly different
than the one taken in this paper.
A common assumption is that the data is fixed,
and various preprocessing methods can be applied.
Most well known of this line of research is bound to be
the work by Fagin and others related to the
threshold algorithm~\cite{fagin2001,gunzer00optimizing,nepal99query},
and its variants, see e.g.~\cite{theobald04topk,das06answering,kumar09topk}.
The idea is to sort each column of the relation
in decreasing order of its values as a preprocessing step.
As a consequence tuples that have a high total score
should appear sooner in the sorted lists.
We cannot make use of this approach,
as it requires reading all values of the input relation
in order to do the sorting,
or alternatively data sources that
directly provide the columns in sorted order,
neither of which we do not have at our disposal.

Considerably more related to this paper
is the work of Marian et.~al.~\cite{marian04tods},
where the problem of aggregating several
web-based data sources is considered.
They too
assume that probing values from the relation(s)
is time-consuming and therefore
the algorithm should aim to minimize
the total execution time of the query.
Another important reference to the current work
is the MPro algorithm discussed by Hwang and Chang \cite{mpro},
to which the algorithm in \cite{marian04tods}
is closely related.
The crucial difference to our work is that
neither \cite{marian04tods} nor \cite{mpro}
consider a similar use of training data,
and require an exact top-$k$ list as the result.

In addition to top-$k$ processing,
we also briefly mention work on
classical database query optimization.
Especially of interest to us is
research on optimizing queries with
\textit{expensive predicates} \cite{hellerstein93,kemper94}.
Of course \cite{mpro} falls to this category as well,
as it considers top-$k$ processing under expensive predicates.
Here the fundamental question concerns
finding a query plan 
to minimize the total execution time
given that
some (restriction) predicates used in the query
are computationally expensive.
These expensive predicates can be e.g.
arbitrary user defined functions.
Our work can be seen in this framework as well.
We consider the processing of a query
(a type of SELECT)
that must return
all rows of the hidden relation that
belong to an approximate top-$k$ set
defined by the given scoring function.
We can assume that
reading one entry of the hidden relation
on a particular row
corresponds to evaluating one expensive predicate
for this row.
While the value of an entry does not directly
specify whether or not the row belongs to the top-$k$ set,
the algorithm that we propose later
in Section~\ref{sect:alg_probab}
gives a probabilistic estimate of this
based on the known entries of a row.


\subsubsection{Approximate nearest neighbors}

Since our ranking is based on the inner product
between the rows and a scoring vector,
the top-$k$ set is equivalent to the
set of $k$ nearest neighbors of the scoring vector
if everything is normalized to unit length.
Algorithms for $k$-NN queries
(in high-dimensional Euclidean spaces)
have been widely studied.
In particular, papers related to
approximate nearest neighbor search \cite{indyk04survey}
are of interest in the context of our work.
The usual approach in these
is to reduce the number of
required distance computations
by preprocessing
the set of points that is being queried.
In locality sensitive hashing \cite{gionis99lsh}
the underlying relation is indexed
using a number of hash functions so that
collisions indicate close proximity.
A query vector is only compared to vectors
mapped to the same bucket by the hash function.
Kleinberg \cite{klei99knn}
takes a similar approach
but uses random projections instead of hash functions.
A somewhat different preprocessing technique
are low dimensional embeddings \cite{ailon06} that
aim to speed up the processing by
representing the set of points in
a lower dimensional space
where distance computations
can be carried out faster.
Singitham et.~al.~\cite{singitham04} propose
a solution based on clustering of the database,
where the query vector is only compared to points
that reside in clusters whose centroid is close
to the query vector.
Recently Goel et.~al.~\cite{goel08}
propose another technique based on clustering that uses
the query-distribution together with
a variant of the threshold algorithm \cite{fagin2001}.

However, the basic assumption in
\cite{gionis99lsh,klei99knn,singitham04,ailon06,goel08}
and related literature is that
the data being queried
is known a-priori so that
indexing techniques can be applied
to quickly find points that are close to an
arbitrary query vector.
To see the problem that
we discuss in this paper
as $k$-NN search,
we have to turn the setting upside down,
so that the query vector
(i.e.~our scoring weights)
is fixed,
and the set of points 
(i.e.~the rows of our hidden relation)
are sampled from a known distribution.
Moreover, an elementary property of our problem
are the costs associated with
reading values from the hidden relation.
Such assumptions are
to the best of our knowledge
not made in any of the existing work on $k$-NN search.

\subsubsection{Learning theory}

Unlike methods for approximate nearest neighbor search,
some models in computational learning theory take
costs for accessing input items into account
\cite{charikar02,kaplan05lac,cicalese05,greiner06}.
In general this line of work considers
ways to evaluate a (boolean) function when
the inputs 
are obtained only by paying a price.
An algorithm is given
a representation of the function
(e.g.~a boolean circuit),
and the costs associated with each input.
The algorithm must learn the value of the function
while keeping the cost of the process low.
In the simplest case the algorithm
is merely an ordering of the variables.
That is, the function is evaluated by reading
values of the variables
according to a specified order.
This approach is studied
e.g.~in \cite{kaplan05lac,greiner06}, while
more complex algorithms are considered
in \cite{charikar02} and \cite{cicalese05}.

On a high level our problem is similar.
We too are concerned with
evaluating a function
(the top-$k$ query)
while trying to minimize the overall cost.
Especially the problem
of finding a good order in which to read the attribute values
that we discuss in
Section~\ref{section:schedule}
is related.
This order is important
as it can have a considerable effect
on the performance of our approach.
It would be of interest to see if
any of the previous results
\cite{charikar02,kaplan05lac,cicalese05,greiner06}
can be applied in this case,
but we consider this to be worth a
discussion of its own
in future work.

\subsection{Our contributions}

We conclude this section with
a structure and
summary of the contributions of this paper.
\begin{itemize}
\item Section~\ref{sect:defs}:
We describe (to the best of our knowledge)
a novel top-$k$ search problem.
The main characteristic of the problem is
that instead of applying
preprocessing techniques on the items that
we are ranking,
we have a sample from the same distribution
to be used as training data for machine learning methods.
\item Section~\ref{sect:alg_probab}:
We propose a simple algorithm that
finds a set of $k$ rows
from a given matrix
with a high score according to a fixed linear scoring function.
The algorithm uses two parameters.
The first parameter is a threshold value that
is used to prune items that are unlikely to belong to the top-$k$ set.
The second parameter is an ordering of the attributes.
\item Section~\ref{section:alpha}:
We propose an algorithm for learning
a good value of the threshold parameter
based on training data.
\item Section~\ref{section:schedule}:
We propose an algorithm for learning
a good ordering of the attributes
based on training data.
\item Section~\ref{section:experiments}:
We conduct a set of experiments to
demonstrate the performance of our algorithm(s).
We compare our algorithm
to a simple baseline,
and another algorithm presented previously in \cite{mpro}.
\end{itemize}

\section{Basic definitions}
\label{sect:defs}

\paragraph{Input matrix}
Let $\X$ be an $n \times m$ matrix,
an element of which is denoted by $\X_{ij}$.
The $i$th row of $\X$, denoted $\X_{i \cdot}$,
represents the $i$th item that we are ranking.
Let $A_1, \ldots, A_m$ be a set of $m$ attributes.
The values of attribute $A_j$ appear on
the $j$th column of $\X$, denoted $\X_{\cdot j}$.
For the rest of this paper we will
assume that $\X_{ij} \in \mathbb{R}_0^+$ for all $i$ and $j$.
That is, all entries of $\X$ are nonnegative real numbers.

\paragraph{Cost of a query}
To each attribute $A_j$
is associated
a \textit{cost} $C(A_j)$
that represents the effort of examining the value $\X_{ij}$.
We assume that this computation
is equally hard for all cells in $\X_{\cdot j}$.
The cost of a top-$k$ query is simply
the sum of the costs of all entries that
our algorithm has to inspect in order to
return its output,
normalized by the cost of the trivial algorithm
that computes all entries of $\X$.
We have thus
\begin{equation}
\label{eqn:cost}
\cost_k(\X) = \frac{\sum_{ij} C(A_j)\mathcal{I}\{\X_{ij} \mbox{ is inspected}\}}{n \sum_jC(A_j)},
\end{equation}
where $\mathcal{I}\{X\}$ is $1$ if
the statement $X$ is true, and $0$ otherwise.

\paragraph{Scoring and top-$k$ sets}
Let $\w = (\w_1, \ldots, \w_m)$
be a (row) vector of weights.
The \textit{prefix} of a vector,
denoted $\w_{1:h}$,
is a $h$-dimensional vector
consisting of the elements $\w_1, \ldots, \w_h$.
Likewise,
we denote by $\X_{i, 1:h}$
the prefix of the $i$th row of $\X$.
The \textit{prefix score}
of the $i$th item
is given by the product $\X_{i, 1:h}\w_{1:h}^T$.
When we have $h = m$,
the prefix score is 
the \textit{full score} $\X_{i \cdot}\w^T$.
The \textit{exact} top-$k$ set of $\X$ given $\w$,
denoted $\topk^k_\w(\X)$,
consists of the indices of the $k$ items
with the highest full scores.
More formally, we have
\[
\topk^k_\w(\X) = \bigl\{ i \; \big| \; |\{ i' \neq i : \X_{i' \cdot}\w^T > \X_{i \cdot}\w^T \}| < k \bigr\}.
\]

\paragraph{The schedule}
All algorithms that we consider in this paper
have the common property
that attributes on row $\X_{i \cdot}$ are
examined sequentially in a certain order,
and this order is the same for all $i$.
That is,
the entry $\X_{ij}$ will be read only
if all entries $\X_{ij'}$, with $j' < j$,
have already been read.
We adopt the terminology used in \cite{mpro}
and call this order the \textit{schedule}.
This resembles the order
in which a database system would
apply selection predicates
in a serial (as opposed to conditional)
execution plan.
However, in our case
the benefit of using one schedule over another
is not associated with selection efficiency,
but having better estimates of the full score
given a prefix score.
In Section~\ref{section:schedule} we
discuss a number of simple baseline schedules,
and also present a method for finding
a good schedule 
using training data.
Different choices for the schedule are
compared in the empirical section.

\paragraph{Accuracy of an approximate result}
The algorithm we propose in this paper
is not guaranteed to return the exact top-$k$ set.
Denote by $\atopk^k_\w(\X)$
the $k$ highest scoring items returned by an
inexact top-$k$ algorithm.
We report the accuracy of such an
approximate top-$k$ list
as the fraction of items
in $\atopk^k_\w(\X)$ that also
belong to the exact set $\topk^k_\w(\X)$.
More formally, we have
\begin{equation}
\label{eqn:accuracy}
\acc_k(\X) = \frac{|\topk^k_\w(\X) \cap \atopk^k_\w(\X)|}{k}.
\end{equation}

\paragraph{Problem setting}
The basic objective
of this paper
is to devise an algorithm
that finds an approximate top-$k$ set
with high accuracy at a low cost.
This can be formalized as a computational problem
in a number of ways.
The simplest approach is to assume there
is an external constraint in the form of a
budget $x$ on the costs,
or a requirement $y$ on the accuracy.
Then we could devise algorithms that
maximize accuracy given that the cost can be at most $x$,
or minimize the cost given that the accuracy has to be at least $y$.
The approach we take in this paper is more pragmatic, however.
We discuss an algorithm that uses two parameters,
both of which affect accuracy and cost.
While we give no analytical guarantees about the performance,
we develop methods to systematically find good values
for these parameters,
where goodness is measured by using accuracy and cost
as defined above.

\section{Baseline algorithms}


We compare the algorithm presented in this paper
with two baseline methods.
The first one makes use of a simple branch-and-bound strategy,
while the second one is the MPro algorithm \cite{mpro}.
Unlike the proposed algorithm, 
they do not need training data
and can be applied in a traditional top-$k$ setting.
Instead they rely on upper bounds,
denoted $U(A_i)$,
for the values of each attribute $A_i$.
These can be based either
on prior knowledge of the attribute domains,
or alternatively on a separate training data.
Combined with the prefix score of the row $\X_{i \cdot}$,
we can use these to upper bound the full score of $\X_{i \cdot}$.
More formally, denote by $U_h(i)$
an upper bound for the full score $\X_{i \cdot}\w^T$
given the prefix score $\X_{i, 1:h}\w^T$
and the upper bounds for the attributes outside the prefix.
We have thus
\[
U_h(i) = \X_{i, 1:h}\w_{1:h}^T + \sum_{i=h+1}^m U(A_i).
\]

\subsection{Simple upper bounding}
\label{sect:alg_ub}

A very straightforward approach to our top-$k$ problem
is the following:
consider the upper bound $U_h(i)$ for the row $i$
after computing the values in a prefix of length $h$.
If this upper bound is below the full score of the
lowest ranking item of the current top-$k$ list,
we know that $\X_{i \cdot}$ can not belong to the final top-$k$ list.
Therefore it is not necessary to compute the remaining values,
and we can skip the row.

To apply this heuristic, we need to first get a
candidate top-$k$ set.
This we obtain by reading
all values of the first $k$ rows of $\X$,
and computing their full scores.
Denote by $\delta$ the lowest score in the current top-$k$ set.
For the remaining rows of $\X$,
we start computing the prefix score,
and each time a new attribute is added to the prefix,
we check the value of $U_h(i)$.
If it is below the current value of $\delta$,
we skip the rest of $\X_{i \cdot}$,
if not, we examine the value of the next attribute.
Once all attributes for a row have been computed,
we know its full score,
and can determine whether or not
it enters the current top-$k$ list.
If it does, we update $\delta$ accordingly.
In the remaining of this paper we call this algorithm
the \ubalg algorithm.

The performance of this method depends on
how rapidly $\delta$ reaches a level that
leads to efficient pruning.
Obviously when $\delta$ is small the value of $U_h(i)$
will always be larger.
We can improve the efficiency of the method
with the following heuristic:
Note that the value of $A_1$
is always computed for every row.
This is because $U_0(i)$ is always
larger than any possible $\delta$,
so nothing will be pruned at this point.
We can thus compute all values in the column $\X_{\cdot 1}$,
and rank the rows of $\X$ in decreasing order of this
without sacrificing anything in the final cost.
After sorting our initial top-$k$ list
will contain rows that have a high value at least
in the first attribute.
They are thus somewhat more likely
to have a high full score
than randomly chosen rows.


\subsection{The MPro algorithm}
\label{sect:alg_mpro}

The MPro algorithm of \cite{mpro} can be seen as the
well known A$^*$ algorithm \cite[p.97ff]{russellnorvig}
adopted for the top-$k$ query problem.
Like the \ubalg algorithm,
it also computes entries in $\X$ in a left-to-right fashion,
i.e., the algorithm does not access $\X_{ij}$ unless
the value $\X_{ij'}$ has been read for all $j' < j$.
For every row $\X_{i \cdot}$ the algorithm
maintains the upper bound $U_h(i)$.
The rows are stored in a priority queue $Q$
with $U_h(i)$ as the key, i.e.,
the the first row in the queue is the
one with the highest upper bound.
The algorithm pops rows from $Q$ one by one,
computes the next unknown entry,
updates the upper bound and inserts the row back into $Q$,
or outputs it as a member of the top-$k$ set if
all values have been computed.
When the output size reaches $k$, the algorithm terminates.
As with the \ubalg algorithm,
as a first step
the value of the attribute $A_1$
is computed for all rows
to compute the initial values of the upper bounds.
These are used to initialize $Q$.
In the remaining of this paper,
we call this algorithm the \mproalg algorithm.

\section{An algorithm based on prior knowledge}
\label{sect:alg_probab}

In this section we describe a method
that finds $k$ high scoring rows
of a given matrix $\X$ using a fixed scoring vector $\w$.
A difference to the baseline methods is that
the algorithm requires prior knowledge of
the distribution of the values in $\X$.
In practice this means we need
training data in form of one or several matrices $\X'$
that are drawn from the same distribution as $\X$.
The algorithm has two parameters
that can be adjusted to tune its performance.
We also provide algorithms for
finding good values for these parameters
from training data.

\subsection{Algorithm outline}

On a high level
the algorithm is based on the
same basic principle
as the \ubalg algorithm.
We scan the rows of $\X$ one by one and
incrementally compute the prefix score for each row.
This is done
until we can discard the remaining entries of the row
based on some criterion,
or until we have computed the full score.
If we decide to skip the row based on a prefix score,
we never return to inspect
the remaining entries of the same row.
However, unlike with the \ubalg or \mproalg algorithms,
we are not using simple upper bounds
for the remaining attributes.
Instead we use the training data $\X'$
to learn a model that allows us to estimate
the probability that the current row
will enter the current top-$k$ set
given the prefix score.
If this probability is below a given threshold value,
we skip the row.

Suppose that we currently have
a candidate set of top-$k$ rows.
Denote by $\delta$
the lowest score in the candidate set,
and let $\X_{i \cdot}$ be the row that
the algorithm is currently considering.
Given a prefix score of $\X_{i \cdot}$,
we can give an estimate for the full score $\X_{i \cdot}\w^T$,
and make use of this together with $\delta$ to decide
whether or not it is worthwhile to
compute the remaining,
still unknown values of $\X_{i \cdot}$.
More precisely, we want to estimate the probability
that $\X_{i \cdot}$ would enter the current top-$k$ set
given the prefix score $\X_{i,1:h}\w_{1:h}^T$, that is
\begin{equation}
\label{eqn:probab}
\Pr\left( \X_{i \cdot}\w^T > \delta \; | \; \X_{i,1:h}\w_{1:h}^T \right).
\end{equation}
If this probability is very small,
say, less than 0.001,
it is unlikely that $\X_{i \cdot}$ will ever enter the top-$k$ set.
In this case we can skip $\X_{i \cdot}$
without computing values of its remaining attributes.
Of course this strategy may lead to errors,
as in some cases
the prefix score may give poor estimates of the full score,
which in turn causes the probability estimates to be incorrect.
The details of estimating Equation~\ref{eqn:probab}
are discussed in Section~\ref{section:probab}.

\begin{algorithm}[t]
\caption{(The \probabalg algorithm)}
\label{alg:probab}
Input: the $n \times m$ matrix $\X$, parameter $\alpha \in [0,1]$\\
Output: an approximate top-$k$ set
\begin{algorithmic}[1]
\STATE Compute all values in column $\X_{\cdot 1}$ and sort the rows of $\X$ in decreasing order of this value.
\STATE $\C_k \leftarrow \{ \X_{1 \cdot}, \ldots, \X_{k \cdot} \}$
\STATE $\delta \leftarrow \min_{\x \in \C_k} \{ \x\w^T\}$
\FOR{$i=k+1$ to $n$}
  \STATE $h \leftarrow 1$
  \WHILE{$h < m$ and $\Pr( \X_{i \cdot}\w^T > \delta \; | \; \X_{i, 1:h}\w_{1:h}^T ) > \alpha$}
    \STATE $h \leftarrow h + 1$
    \STATE Compute the value $\X_{ih}$.
  \ENDWHILE
  \IF{$h = m$ and $\X_{i \cdot}\w^T > \delta$}
    \STATE $\C_k \leftarrow \C_k \setminus \arg\min_{\x \in \C_k}\{ \x\w^T\}$
    \STATE $\C_k \leftarrow \C_k \cup \X_{i \cdot}$
    \STATE $\delta \leftarrow \min_{\x \in \C_k}\{ \x\w^T \}$
  \ENDIF
\ENDFOR
\STATE \textbf{return} $\C_k$
\end{algorithmic}
\end{algorithm}

An outline of the \probabalg algorithm we propose
is given in Algorithm~\ref{alg:probab}.
It uses a parameter $\alpha$ that
determines when remaining entries
on the row $\X_{i \cdot}$ are to be skipped.
Whenever we have
$\Pr( \X_{i \cdot}\w^T \geq \delta \; | \; \X_{i,1:h}\w_{1:h}^T ) < \alpha$
we proceed with the next row.
Selecting an appropriate value of $\alpha$
is discussed in Section~\ref{section:alpha}.
As with the baseline algorithms,
we also need an order, the schedule,
in which to process the attributes.
This is the 2nd parameter of our algorithm.
In Section~\ref{section:schedule} 
we describe a number of simple baseline schedules,
and also propose a method that uses training data
to learn a good schedule for the \probabalg algorithm.

\subsection{Estimating the probabilities}
\label{section:probab}

The most crucial part of our algorithm
is the method for estimating the probability
$\Pr\bigl( \X_{i \cdot}\w^T > \delta \; | \; \X_{i,1:h}\w_{1:h}^T \bigr)$.
In short, the basic idea is to
estimate the distribution of $\X_{i \cdot}\w^T$
given the prefix score $\X_{i, 1:h}\w_{1:h}^T$.
We do this by
learning regression models that predict
the parameters of this distribution
as a function of the prefix score.
Together with $\delta$ the desired probability can be
found out using this distribution.
The details of this are discussed next.


The basic assumption of this paper
is that the distribution of the full score 
$\X_{i \cdot}\w^T$ given a fixed prefix score
is Gaussian.
We acknowledge that this may not be true in general.
However, according to
the central limit theorem,
as the number of attributes increases,
their sum approaches a normal distribution
as long as they are independent.
(The attributes need not follow the same distribution
as long as they are bounded,
see the Lindeberg theorem \cite[page 254]{feller50}.)
Of course the attributes may not be independent,
and also their number may not be large enough
to fully warrant this argument in practice.
Nonetheless, we consider this a reasonable first step.

By convention, we denote the parameters of the normal distribution
by $\mu$ and $\sigma$,
where $\mu$ is the mean
and $\sigma$ the standard deviation.
Furthermore, we assume that both $\mu$ and $\sigma$
depend on the prefix score,
and we must account for prefixes of different lengths.
Denote by $s_h$ a prefix score that
is based on the first $h$ attributes.
The assumption is that
$\X_{i \cdot}\w^T \sim N\bigl(\mu(s_h), \sigma(s_h)\bigr)$.
Once we have some estimates for $\mu(s_h)$ and $\sigma(s_h)$,
we simply look at the tail of the distribution
and read the probability of
$\X_{i \cdot}\w^T$ being larger than a given $\delta$.
To learn the functions
$\mu(s_h)$ and $\sigma(s_h)$
we use training data.
For every possible prefix length $h$,
we associate the prefix score of the row $\X_{i \cdot}$
with the full score of $\X_{i \cdot}$.
That is,
our training data consists of
the following set of
(``prefix score'', ``full score'') pairs for every $h$:
\begin{equation}
\label{eqn:training1}
\mathcal{X}_h = \{ (\X_{i, 1:h}\w_{1:h}^T, \X_{i \cdot}\w^T ) \}_{i=1}^n.
\end{equation}

Now we have to estimate $\mu(s_h)$ and $\sigma(s_h)$.
One approach is to use binning.
Given $s_h$ and $\mathcal{X}_h$, we could compute the set
\[
B(s_h) = \{ b \; | \; a \in \Bin(s) \wedge (a,b) \in \mathcal{X}_h \}
\]
that contains full scores of objects
that have a  prefix score belonging to the same bin as $s_h$.
The bins are precomputed in advance by some suitable technique.
Now we can define $\mu(s_h)$ and $\sigma(s_h)$
simply as their standard estimates in $B(s_h)$.
This approach has some drawbacks, however.
First, we need to store the sets $\mathcal{X}_h$ for every $h$.
This might be a problem if $n$ and $m$ are very large.
Whereas if $n$ is small,
we either have to use large bins,
which leads the estimates being only coarsely connected to $s_h$,
or use narrow bins with only a few examples in each,
which will also degrade the quality of the estimators.

To remedy this
we use an approach based on kernel smoothing \cite{kernelbook}.
Instead of fixed bins,
we consider all of $\mathcal{X}_h$
when computing an estimate of $\mu(s_h)$ or $\sigma(s_h)$.
The idea is that
a pair $(a,b) \in \mathcal{X}_h$ contributes
to the estimates with a weight that
depends on the distance between
the prefix scores $a$ and $s_h$.
The pair contributes a lot if $a$ is close to $s_h$,
and only a little (if at all) if the distance is large.
Denote by $K: \mathbb{R} \times \mathbb{R} \rightarrow \mathbb{R}$
a \textit{kernel function}.
For the rest of this paper we let
\begin{equation}
K(x,y) = e^{- \frac{|| x - y ||}{\beta}},
\end{equation}
where $\beta$ is a parameter.
Other alternatives could be considered as well,
the proposed method is oblivious to
the choice of the kernel function.

Using $K$ we can define
the kernel weighted estimates for
$\mu(s_h)$ and $\sigma(s_h)$.
We let
\begin{equation}
\label{eqn:kernel_mu}
\mu(s_h) = \frac{\sum_{(a,b) \in \mathcal{X}_h}K(a,s_h)b}{\sum_{(a,b) \in \mathcal{X}_h}K(a,s_h)},
\end{equation}
that is, any full score $\X_{i \cdot}\w^T$
contributes to $\mu(s_h)$
with the weight $K(\X_{i,1:h}\w_{1:h}^T, s_h)$.
The nice property of this approach is
that it can be also used to estimate
the standard deviation of the full score at $s_h$
by letting
\begin{equation}
\label{eqn:kernel_sigma}
\sigma(s_h) = \sqrt{\frac{\sum_{(a,b) \in \mathcal{X}_h}K(a,s_h)b^2}{\sum_{(a,b) \in \mathcal{X}_h}K(a,s_h)} - \mu(s_h)^2}.
\end{equation}
The above equation is a simple variation
of the basic formula $\Var[X] = \E[X^2] - \E[X]^2$,
where the kernel function is taken into account.

One problem associated with
kernel smoothing techniques in general is
the \textit{width} of the kernel
that in this case is defined by
the parameter $\beta$.
Small values of $\beta$
have the effect that the prefix score
$a$ of a pair $(a,b) \in \mathcal{X}_h$ must be very close to $s_h$
for the full score $b$ to contribute anything to the final estimates.
Larger values have the opposite effect,
even points that are far away from $s_h$ will influence the estimates.
Selecting an appropriate width for the kernel is not trivial.
We observed that setting $\beta$
to one 5th of the standard deviation of the prefix scores for $h$
gives good results in practice.

While this technique
lets us avoid some of the problems
related to the binning approach,
it comes at a fairly high computational cost.
We have to evaluate the kernel $n$ times
to get estimates for $\mu(s_h)$ and $\sigma(s_h)$ for one $s_h$.
These estimates must be computed potentially
for every possible prefix of every row in $\X$.
This results in $O(n^2m)$ calls to $K(x,y)$
for one single query
(assuming both the training data and
the input matrix have $n$ rows),
which clearly does not scale.
Hence, we introduce approximate estimators
for $\mu(s_h)$ and $\sigma(s_h)$
that are based on simple linear regression models.
This way we do not need to
evaluate $K(x,y)$ at query time at all.
We let
\begin{equation}
\label{eqn:muhat}
\hat{\mu}(s_h) \sim q_1^\mu s_h + q_0^\mu,
\end{equation}
and
\begin{equation}
\label{eqn:sigmahat}
\hat{\sigma}(s_h) \sim q_1^\sigma s_h + q_0^\sigma.
\end{equation}
The parameters $q_0^\mu$, $q_1^\mu$,
$q_0^\sigma$, and $q_1^\sigma$ are the standard estimates
for linear regression coefficients
given the sets
\begin{equation}
\label{eqn:training3}
T_\mu = \{ (\X_{i, 1:h}\w_{1:h}, \mu(\X_{i, 1:h}\w_{1:h}) \}_{i=1}^n,
\end{equation}
and
\begin{equation}
\label{eqn:training4}
T_\sigma = \{ (\X_{i, 1:h}\w_{1:h}, \sigma(\X_{i, 1:h}\w_{1:h}) \}_{i=1}^n,
\end{equation}
where $\mu(\X_{i, 1:h}\w_{1:h})$
and $\sigma(\X_{i, 1:h}\w_{1:h})$
are based on equations
\ref{eqn:kernel_mu} and \ref{eqn:kernel_sigma},
respectively.
We thus compute the kernel estimates only for the training data.
Given $T_\mu$ and $T_\sigma$ we
learn linear functions that are used at query time
to estimate the parameters of the normal distribution 
that we assume the full scores are following.

Our method for estimating the probability
$\Pr\bigl( \x\w^T > \delta \; | \; \x_{1:h}\w_{1:h}^T \bigr)$
can be summarized as follows:
\begin{enumerate}
\item Given a training data
(a matrix $\X$ with all entries known),
compute for each row the full score,
and associate this with the prefix scores
for each possible prefix length $h$.
That is, for each $h$ compute the set
$\mathcal{X}_h$ as defined in Equation~\ref{eqn:training1}.
\item Using the definitions for $\mu(s_h)$ and $\sigma(s_h)$
given in equations \ref{eqn:kernel_mu} and \ref{eqn:kernel_sigma},
compute the sets $T_\mu$ and $T_\sigma$ defined in
equations \ref{eqn:training3} and \ref{eqn:training4},
respectively.
\item Learn the models in
equations \ref{eqn:muhat} and \ref{eqn:sigmahat}
by fitting a regression line
to the points in $T_\mu$ and $T_\sigma$, respectively.
\item At query time,
use the cumulative density function of
$N\bigl(\hat{\mu}(\X_{i, 1:h}\w_{1:h}^T), \hat{\sigma}(\X_{i, 1:h}\w_{1:h}^T)\bigr)$
to estimate the probability of
$\X_{i \cdot}\w^T$ being larger than $\delta$.
\end{enumerate}

\section{Parameter selection}

In this section we discuss
systematic methods for
choosing the parameters required by
the algorithm presented above.

\subsection{Choosing the right $\alpha$}
\label{section:alpha}

We start by describing
a method for learning an ``optimal'' value of $\alpha$
given training data $\X$.
This can be very useful,
since setting
the value of $\alpha$ too low
will decrease
the performance of Algorithm~\ref{alg:probab}
in terms of the cost.
When $\alpha$ increases,
the algorithm will clearly prune more items.
This leads both to a lower cost and a lower accuracy.
Conversely, when alpha decreases,
the accuracy of the method increases,
and so does the cost as less items are being pruned.
The definitions of accuracy and cost in
equations \ref{eqn:accuracy} and \ref{eqn:cost}, respectively,
thus depend on $\alpha$.
We denote by $\acc_k(\X, \alpha)$ and $\cost_k(\X, \alpha)$
the accuracy and cost
attained by the \probabalg algorithm
for a given value of $\alpha$.

Due to the trade-off between cost and accuracy,
we should set $\alpha$ as high (or low) as possible
without sacrificing too much in accuracy (or cost).
While a very conservative estimate for $\alpha$,
say $0.001$, is quite likely
to result in a high accuracy,
it can perform sub-optimally in terms of the cost.
Maybe with $\alpha = 0.05$ we obtain an almost equally high accuracy
at only a fraction of the cost.

Consider a coordinate system
where we have accuracy on the x-axis
and cost on the y-axis.
In an ideal setting we would have
a accuracy of $1$ at zero cost,
represented by the point at $(1,0)$
on this accuracy-cost plane.
Obviously this is not attainable in reality,
since we always have to inspect some of the entries of $\X$,
and this will lead to a nonzero cost.
But we can still define
the optimal $\alpha$ in terms of this point.

\begin{definition}
\label{def:alpha}
Let
\[
\dist_k(\X, \alpha) = || \bigl(\acc_k(\X, \alpha), \cost_k(\X, \alpha) \bigr) - (1,0) ||.
\]
The optimal $\alpha^*$ given the matrix $\X$ satisfies
\[
\alpha^* = \arg\min_{\alpha \in [0,1]} \dist_k(\X, \alpha),
\]
where $||\cdot||$ denotes the Euclidean norm.
\end{definition}
That is, we want to find an $\alpha$
that minimizes the distance to the point $(1,0)$
on the accuracy-cost plane.
Clearly this is a rather simple definition.
It assigns equal weight to accuracy and cost,
even though we might prefer one over the other,
depending on the application.
However, modifying the definition to take
such requirements into account is easy.

Next we discuss how to find $\alpha^*$.
In the definition we state that it
has to belong to the interval $[0,1]$.
However, first we observe that there exists
an interval $[\alpha_{\min}, \alpha_{\max}]$,
so that when $\alpha \leq \alpha_{\min}$
we have $\acc_k(\X, \alpha) = 1$,
and when $\alpha \geq \alpha_{\max}$
we have $\acc_k(\X, \alpha) = 0$.
Clearly the the interesting $\alpha$
in terms of Definition~\ref{def:alpha}
lies in $[\alpha_{\min}, \alpha_{\max}]$.
We can analyze the values in this interval even further.
Consider the following set of possible values for $\alpha$:
\begin{equation}
\label{eqn:qx}
Q(\X) = \{ \min_h \Pr( \x\w^T > \delta \; | \; \x_{1:h}\w_{1:h}^T) \}_{\x \in \topk^k_\w(\X)},
\end{equation}
where $\delta = \min_{x \in \topk^k_\w(\X)} \{ \x\w^T \}$.
That is, 
for each $\x \in \topk^k_\w(\X)$,
$Q(\X)$ contains the value $a$ 
so that when $\alpha > a$, Algorithm~\ref{alg:probab} will prune $\x$.
More precisely,
if we order the values in $Q(\X)$ in ascending order,
and let $a_i$ denote the $i$th value in this order,
we know that when $\alpha \in [a_i, a_{i+1})$
the algorithm will prune exactly $i$ rows of the
correct top-$k$ set of $\X$.
(Assuming that all $a_i$ are different.)
By letting $\alpha$ vary from $a_1 = \alpha_{\min}$ to $a_k < \alpha_{\max}$,
$\acc_k(\X,\alpha)$ decreases from $1$ to $1/k$ in steps of $1/k$.
Likewise, $\cost_k(\X,\alpha)$ decreases as $\alpha$ increases.
Now we can systematically express $\cost_k(\X,\alpha)$
as a function of $\acc_k(\X,\alpha)$,
since each $a \in Q(\X)$ is associated with
a certain accuracy.

This makes finding the optimal $\alpha$ easy.
We solve the optimization problem of Definition~\ref{def:alpha}
by only considering values in $Q(\X)$.
In fact, we can show that
an $\alpha^*$ obtained this way
is the same as the one we would obtain by
having the interval $[0,1]$ as the feasible region.
\begin{lemma}
Let $\alpha^* = \arg\min_{\alpha \in [0,1]} \dist_k(\X, \alpha)$. We have
$\alpha^* \in Q(\X)$,
where $Q(\X)$ is defined as in Equation~\ref{eqn:qx}.
\end{lemma}
\begin{proof}
We show that for all $\alpha$s that lie between
any two adjacent values in $Q(\X)$,
the distance $\dist_k(\X, \alpha)$
is larger than when $\alpha$ is chosen from $Q(\X)$.
Consider any $a_i$ and $a_{i+1}$ in $Q(\X)$.
We show that
within the interval $[a_i, a_{i+1}]$
the distance $\dist_k(\X, \alpha)$ is minimized
for either $\alpha = a_i$ or $\alpha = a_{i+1}$.
As $\alpha$ increases
from $a_i$ to $a_i + \epsilon$ for some small $\epsilon > 0$,
$\acc_k(\X, \alpha)$ decreases by $1/k$,
and $\dist_k(\X, \alpha)$ increases by
$\bigl(\dist_k(\X, a_i + \epsilon) - \dist_k(\X, a_i)\bigr) = \Delta_1 > 0$.
When we further increase $\alpha$ from $a_i + \epsilon$ to $a_{i+1}$,
$\acc_k(\X, \alpha)$ stays the same,
but $\cost_k(\X, \alpha)$ may decrease.
Therefore, $\dist_k(\X, \alpha)$ decreases until $\alpha = a_{i+1}$.
We let
$\bigl( \dist_k(\X, a_i + \epsilon) - \dist_k(\X, a_{i+1}) \bigr) = \Delta_2 > 0$.
If $\Delta_1 > \Delta_2$,
we have $\dist_k(\X, a_i) < \dist_k(\X, a_{i+1})$,
otherwise
$\dist_k(\X, a_i) > \dist_k(\X, a_{i+1})$.
\end{proof}



\subsection{Choosing a schedule}
\label{section:schedule}

So far we have not considered
the order,
the \textit{schedule},
in which the columns of $\X$ should be processed.
This order has a considerable impact
on the performance of the algorithms.
Processing the attributes in a certain order
will lead to a tighter upper bound on the full score
in case of
the \ubalg and \mproalg algorithms.
With the \probabalg algorithm
the probability estimates will be more accurate
with some permutations of the attributes than others.
A similar problem
was considered in \cite{mpro}
for the \mproalg algorithm.
The approach is different, however,
as training data is not used
and an optimal schedule must be found at query time.

\subsubsection{Baseline schedules}

Given the ranking vector $\w$,
and the cost $C(A_i)$ for each attribute $A_i$,
we consider four simple baselines
for the schedule:
\begin{itemize}
\item[A:] Read the attributes in \textit{random order}.
This is the simplest possible way of choosing a schedule.
We pick a random total order of $m$ items uniformly
from the set of all permutations and use this as the schedule.
\item[B:] Read the attributes in
\textit{decreasing order of the absolute values in} $\w$.
This can be motivated by
the fact that attributes with a larger weight
(the important attributes)
will have a bigger impact on
the full score $\w^T\X_{i \cdot}$.
In some cases we might have a fairly accurate
estimate of $\w^T\x$ already after a very short prefix
of the row $\X_{i \cdot}$ has been computed
This in turn will lead to better pruning,
since the estimates of the probability
$\Pr\bigl( \w^T\X_{i \cdot} > \delta \; | \; \w_{1:h}^T\X_{i, 1:h} \bigr)$
are more accurate.
The downside of this approach is that
the costs are not taken into account.
It is possible that the important attributes
have almost the same absolute value in $\w$,
but considerably different costs.

\item[C:] Read the attributes in
\textit{increasing order of the cost} $C(A_i)$.
This is based on the assumption that
by computing the ``cheap'' features first,
we might be able to prune objects
without having to look at the expensive attributes at all.
However, this time we may end up
computing a long prefix of $\X_{i \cdot}$,
because it is possible that
some of the ``cheap'' attributes
have a low weight in $\w$,
and thereby do not contribute so much to the full score.

\item[D:] Read the attributes in
\textit{decreasing order of the ratio} $|\w_i|/C(A_i)$.
By this we try to remedy the downsides of
the previous two approaches.
The value of an attribute is
high if it has a large weight in $\w$, and a small cost.
Conversely, attributes with a small weight and a high cost
are obviously less useful.
\end{itemize}

\subsubsection{Learning a schedule from training data}

In addition to the baselines above,
we can also try to find a schedule
by using available information.
In general we want to find a schedule
that minimizes the cost
of finding the top-$k$ set
in the training data.
One difficulty here is the selection of $\alpha$.
The cost of a given schedule $\psi$
depends on the value of $\alpha$,
and the optimal schedule might
be different for different values of $\alpha$.
One option would be to
fix $\alpha$ in advance.
However, we want to avoid this,
because the $\alpha$ we use for finding the schedule
might be different from the $\alpha$
that is used when running Algorithm~\ref{alg:probab}.
(After learning the schedule $\psi$,
we can use the method described in Section~\ref{section:alpha}
to find an optimal value of $\alpha$
\textit{given} $\psi$.)
Another option would be to
simultaneously learn
an optimal schedule $\psi$ and the optimal $\alpha$.
This does not seem trivial, however.
Instead, we take an approach where
we try to find a schedule
that is good
\textit{independent} of the final choice of $\alpha$.

If $\alpha$ were fixed,
we could define a cost for the schedule $\psi$
in terms of $\acc_k(\X, \alpha)$ and $\cost_k(\X,\alpha)$.
However, instead of considering a particular value of $\alpha$,
we define the cost as a sum over
all possible meaningful values of $\alpha$.
Recall that the set $Q(\X)$ (see Equation~\ref{eqn:qx})
contains all ``threshold'' values so that
when $\alpha$ crosses these,
$\acc_k(\X, \alpha)$ decreases by $1/k$.
We define the cost of the schedule $\psi$ given $\X$ as
\begin{equation}
\label{eqn:psiscore}
\cost_k(\X, \psi) = \sum_{\alpha \in Q(\X(\psi))} \cost_k(\X(\psi),\alpha),
\end{equation}
where $\X(\psi)$ denotes the matrix $\X$
with permutation $\psi$ applied to its columns.
Note that $\cost_k(\X, \psi)$ can be interpreted as the
``area'' below the curve of $\cost_k(\X(\psi), \alpha)$
in the accuracy-cost plane
for $\alpha \in Q(\X)$.
For example, if the curve corresponding to permutation
$\psi$ is below the curve corresponding to $\psi' \neq \psi$,
we know that independent of $\alpha$,
the schedule $\psi$ always has a smaller cost for
the same value of $\acc_k(\X(\psi), \alpha)$.
The score in Equation~\ref{eqn:psiscore}
is a heuristic that attempts to capture this intuition.
The scheduling problem can thus be expressed as follows:
Given an integer $k$ and the matrix $\X$, find the schedule
\[
\psi^* = \arg\min_\psi \{ \cost_k(\X, \psi) \},
\]
where $\cost(\X, \psi)$ is defined as 
in Equation~\ref{eqn:psiscore}.
In this paper we propose
a simple greedy heuristic
for learning a good schedule.
Denote by a \textit{partial schedule}
a prefix of a full schedule.
The algorithm works by
adding a new attribute to
an already existing
partial schedule.
The attribute that is added
is the best one
among all possible alternatives.

Since we're dealing with partial schedules
that are prefixes of a full schedule,
we can not evaluate $\cost_k(\X, \psi)$
exactly as defined above.
This is because some rows are not
pruned by looking only at their prefix.
However, they may be pruned at some later stage
given a longer prefix.
When evaluating a partial schedule,
we assume that any row that is not pruned 
incur the full cost.
That is,
we must read all of their attributes
before knowing
whether or not they belong to the top-$k$ set.
This means that the cost of a prefix of the final schedule
is an upper bound for the cost of the full schedule.
More formally,
we denote the upper bound by
$\lceil\cost(\X,\psi)\rceil$,
and let
\[
\lceil\cost(\X,\psi)\rceil = \sum_{\alpha \in Q(\X(\psi))} \sum_{i=1}^n \cost(\X_{i \cdot}, \psi, \alpha),
\]
where the row-specific cost is
\begin{equation}
\label{eqn:rowub}
\cost(\x, \psi, \alpha) = \left\{
      \begin{array}{cl}
      \sum_{j=1}^m C(A_j) & \mbox{if } I(\x, \psi, \alpha) = \emptyset,\\
      \sum_{j=1}^{I(\x, \psi, \alpha)} C(A_{\psi(j)})& \mbox{otherwise}.
\end{array}\right.
\end{equation}
Above $I(\x, \psi, \alpha)$ is
the index of the first attribute (according to $\psi$)
that will prune the remaining attributes of $\x$, that is,
\[
I(\x, \psi, \alpha) = \min \{ h \; | \; \Pr( \x\w^T > \delta \; | \; \x_{\psi(1:h)}\w_{\psi(1:h)}^T) < \alpha\}.
\]
For convenience we let $\min \{ \emptyset \}  = \emptyset$.
Equation~\ref{eqn:rowub} simply states that
the cost of a row is the sum of all attribute costs
if the row is not pruned,
otherwise we only pay for the attributes
that are required to prune the row.
\begin{algorithm}[t]
\caption{}
\label{alg:scheduler}
Input: the matrix $\X$, the set of attributes $\mathcal{A}$\\
Output: a permutation $\psi$ of $\mathcal{A}$
\begin{algorithmic}[1]
\STATE $\psi \leftarrow []$
\WHILE{$\mathcal{A} \neq \emptyset$}
  \STATE $A' \leftarrow \arg\min_{A \in \mathcal{A}} \lceil\cost([\psi A], \X)\rceil$
  \STATE $\psi \leftarrow [\psi A']$
  \STATE $\mathcal{A} \leftarrow \mathcal{A} \setminus A'$
\ENDWHILE
\STATE \textbf{return} $\psi$
\end{algorithmic}
\end{algorithm}
The scheduling algorithm, shown in Algorithm~\ref{alg:scheduler},
always appends the attribute to the prefix
that minimizes the upper bound $\lceil\cost(\X,\psi)\rceil$.
We denote by $[\psi A]$ the permutation $\psi$ appended with $A$.

\section{Experiments}
\label{section:experiments}

In the experiments that follow
we compare the performance of our proposed method
with the baseline algorithms
using different schedules.
Our basic criteria for evaluation
are the cost and accuracy measures.
We with to remind the reader
that our notion of accuracy
is not a measure of relevance,
but simply a comparison with the
exact top-$k$ set.
In addition to the baselines described earlier,
it is good to compare the numbers
with a sampling approach,
where we randomly select, say,
50 percent of the rows of the matrix,
and run the trivial algorithm on this.
This will have a cost of $0.5$,
and also the expected accuracy will be $0.5$.
Any reasonable algorithm should
outperform this.

The upper bounds for attribute values
used by the \ubalg and \mproalg algorithms
are based on training data as well.
The upper bound for attribute $A_j$
is the largest value of $A_j$
observed in the training data.
We acknowledge that this is
a rather rudimentary approach,
but we want to study how these
algorithms perform under the same conditions
as the \probabalg algorithm.
In each of the tables that follow,
the numbers in parenthesis
denote the standard deviation of
the corresponding quantity.

\subsection{Datasets}

We conduct experiments on both artificial and real data.
Random data is generated by
sampling each $\X_{ij}$ from
a normal distribution with zero mean
and a unit variance.
To enforce that $\X_{ij} \in \mathbb{R}_0^+$
we replace each entry with its absolute value.
In every experiment we use one random $\X$ as the training data,
and another random $\X$ as the test data.
The results are averages over a number of such training-testing pairs.
Also, the vector $\w$ and the costs $C(A_j)$
are chosen uniformly at random from the interval $[0,1]$.

The real data consists of a set of queries
from a context sensitive Wikipedia search engine \cite{ukkonen08cikm}.
For each query $q$ we have the matrix $\X_q$
where each row corresponds to a document that contains the query term(s).
The documents are represented by 7 features.
We split the data randomly to a training and test part.
The training data consists of 25 matrices,
each corresponding to a set of documents matching a different query.
The test data consists of 100 matrices,
each again corresponding to a different query.
(There is no overlap between
queries in the training data and the test set.)
The training part is used to learn the weight vector $\w$
as described in \cite{ukkonen08cikm}.
Also the algorithms for finding a good schedule
and optimizing the value of $\alpha$ are run on the training data.

The attribute costs $C(A_j)$
for the Wikipedia example
were measured
by computing features for 400 queries.
For each query the result set is restricted to
1000 topmost documents according to one of the features (\textsc{bm25}).
In every case we measure the time spent
computing each feature.
The costs shown in Table~\ref{tbl:wcscosts}
are logarithms of the averages of these.
The numbers are not intended to be fully realistic,
but we consider them reasonable for the purposes of this paper.

\begin{table}
\centering
\caption{Estimated attribute costs and their scoring weights for the Wikipedia data.}
\label{tbl:wcscosts}
\scriptsize
\begin{tabular}{l|rrrrrrr}
         & BM25 & GPR  & TXT   & SUCC & PRED & SPCT & LPR \\
\hline
$C(A_j)$ & 1.43 & 2.23 & 10.02 & 5.49 & 4.06 & 5.42 & 1.72\\
$\w_j$   & 0.047 & 0.003 & 0.636 & 0.479 & 0.353 & 0.008 & 0.588
\end{tabular}
\end{table}

\subsection{Schedule comparison}
\label{sect:exp_schedule}

First we compare the different schedule selection heuristics.
With \mproalg and \ubalg we only use the baseline schedules
A, B, C, and D.
In case of the \probabalg algorithm we also study how
a schedule learned using
the method described in Section~\ref{section:schedule}
compares to the baselines.
With the \probabalg algorithm
we use the method described in Section~\ref{section:alpha}
to learn a good value of $\alpha$.
We also study the effect of the
heuristic described in Section~\ref{sect:alg_ub}.
That is, do we gain anything
by reordering the rows of $\X$
in decreasing order of the value of the
first attribute in the schedule
before running the algorithms.
Note that this affects only
the \ubalg and \probabalg algorithms.
The \mproalg algorithm has this heuristic
built-in as the next element of $\X$ it reads
is selected from a priority queue
that is initialized with 
the upper bounds based on only the first feature.

Upper part of Table~\ref{tbl:schedule_randomdata_cost} shows
the average cost for each algorithm and schedule
for $k=10$ over 50 random inputs
when the row reordering heuristic is in use.
As can be seen, the \probabalg algorithm
outperforms both \ubalg and \mproalg by a clear margin
independent of the choice of the schedule.
When comparing the schedules,
both D (the weight-cost ratio heuristic)
and a learned schedule outperform the others.
The difference between D and a learned schedule is very small.
The bottom part of Table~\ref{tbl:schedule_randomdata_cost}
shows the same quantities for the
\ubalg and \probabalg algorithms
when the rows of the input matrix are not sorted
in decreasing order of the value on the first attribute
(according to the used schedule).
Clearly both algorithms perform considerably
worse in this case.
Hence, with random data the row reordering heuristic is useful.

The accuracies for random data are shown in
Table~\ref{tbl:schedule_randomdata_accuracy}.
Also here the upper and lower parts of the table
show results with and without the row reordering heuristic, respectively.
In general there is a correlation between cost and accuracy;
the more entries of the matrix you inspect,
the more accurate are the results.
In terms of accuracy the \ubalg algorithm
gives the best results,
with nearly 100 percent accuracy in almost every case.
The \probabalg algorithm has an average accuracy of 0.85
with schedule D,
which is a very good result considering that the algorithm
inspected on average only 23 percent of the entries of $\X$.
With the learned heuristic accuracy drops to 0.81, however.

\begin{table}
\centering
\caption{Costs for different schedules using random data ($k=10$) with (top) and without (bottom) the row reordering heuristic.}
\label{tbl:schedule_randomdata_cost}
\scriptsize
\begin{tabular}{l|r@{\hspace{2mm}}r@{\hspace{2mm}}r@{\hspace{2mm}}r@{\hspace{2mm}}r}
           & A           & B           & C           & D           & learned\\
\hline
\ubalg     & 0.83 (0.12) & 0.85 (0.09) & 0.88 (0.09) & 0.88 (0.07) & -      \\
\mproalg   & 0.69 (0.14) & 0.69 (0.15) & 0.60 (0.14) & 0.66 (0.13) & -      \\
\probabalg & 0.44 (0.17) & 0.25 (0.10) & 0.31 (0.15) & 0.23 (0.10) & 0.22 (0.10)\\
\hline
\hline
\ubalg     & 0.86 (0.10) & 0.89 (0.08) & 0.91 (0.07) & 0.91 (0.06) & - \\
\probabalg & 0.56 (0.14) & 0.37 (0.10) & 0.43 (0.12) & 0.35 (0.09) & 0.34 (0.09)
\end{tabular}

\end{table}

\begin{table}
\centering
\caption{Accuracies for different schedules using random data ($k=10$) with (top) and without (bottom) the row reordering heuristic.}
\label{tbl:schedule_randomdata_accuracy}
\scriptsize
\begin{tabular}{l|r@{\hspace{2mm}}r@{\hspace{2mm}}r@{\hspace{2mm}}r@{\hspace{2mm}}r}
           & A           & B           & C           & D           & learned\\
\hline
\ubalg     & 0.87 (0.21) & 0.99 (0.01) & 0.99 (0.01) & 1.00 (0.00) & - \\
\mproalg   & 0.55 (0.21) & 0.88 (0.10) & 0.62 (0.20) & 0.88 (0.11) & - \\
\probabalg & 0.84 (0.14) & 0.84 (0.14) & 0.84 (0.14) & 0.85 (0.16) & 0.81 (0.14)\\
\hline
\hline
\ubalg     & 0.89 (0.19) & 0.99 (0.01) & 0.99 (0.01) & 1.00 (0.00) & - \\
\probabalg & 0.90 (0.09) & 0.91 (0.08) & 0.89 (0.10) & 0.90 (0.10) & 0.88 (0.10)
\end{tabular}

\end{table}

\begin{table}
\centering
\caption{Costs for different schedules using Wikipedia data ($k=10$) with (top) and without (bottom) the row reordering heuristic.}
\label{tbl:schedule_wikidata_cost}
\scriptsize
\begin{tabular}{l|r@{\hspace{2mm}}r@{\hspace{2mm}}r@{\hspace{2mm}}r@{\hspace{2mm}}r}
           & A           & B           & C           & D           & learned\\
\hline
\ubalg     & 1.00 (0.00) & 1.00 (0.00) & 1.00 (0.00) & 0.91 (0.05) & -\\
\mproalg   & 0.95 (0.00) & 0.81 (0.06) & 0.67 (0.00) & 0.82 (0.00) & - \\
\probabalg & 0.43 (0.27) & 0.43 (0.09) & 0.64 (0.28) & 0.43 (0.22) & 0.66 (0.32)\\
\hline
\hline
\ubalg     & 0.76 (0.20) & 0.99 (0.01) & 0.62 (0.23) & 0.99 (0.01) & -\\
\probabalg & 0.50 (0.24) & 0.50 (0.09) & 0.71 (0.24) & 0.50 (0.20) & 0.74 (0.25)
\end{tabular}

\end{table}

\begin{table}
\centering
\caption{Accuracies for different schedules using Wikipedia data ($k=10$) with (top) and without (bottom) the row reordering heuristic.}
\label{tbl:schedule_wikidata_accuracy}
\scriptsize
\begin{tabular}{l|r@{\hspace{2mm}}r@{\hspace{2mm}}r@{\hspace{2mm}}r@{\hspace{2mm}}r}
           & A           & B           & C           & D           & learned\\
\hline
\ubalg     & 1.00 (0.00) & 1.00 (0.00) & 1.00 (0.00) & 1.00 (0.00) & -\\
\mproalg   & 0.76 (0.20) & 0.99 (0.01) & 0.62 (0.23) & 0.99 (0.01) & -\\
\probabalg & 0.63 (0.30) & 0.81 (0.22) & 0.81 (0.25) & 0.83 (0.23) & 0.83 (0.23)\\
\hline
\hline
\ubalg     & 1.00 (0.00) & 1.00 (0.00) & 1.00 (0.00) & 1.00 (0.00) & -\\
\probabalg & 0.64 (0.29) & 0.84 (0.20) & 0.82 (0.25) & 0.83 (0.24) & 0.85 (0.21)
\end{tabular}

\end{table}

Cost and accuracy for the Wikipedia data are shown in Tables
\ref{tbl:schedule_wikidata_cost} and \ref{tbl:schedule_wikidata_accuracy},
respectively.
The numbers are averages over 100 different queries
that belong to the test set.
In terms of the cost the \probabalg algorithm
is again a clear winner.
The best schedules
for \probabalg
are A, B, and D,
with the learned schedule having problems.
When accuracy is considered,
we observe that schedule A
performs considerably worse than the others.
Overall the best choice is D
(order attributes in decreasing order of the ratio $\w_j/C(A_j)$), however.
With this schedule the \probabalg algorithm
attains a accuracy of $0.83$ and
pays only $43$ percent of the maximum cost.
As with random data,
the costs increase for \probabalg when
the row reordering heuristic is not used.
Interestingly \ubalg performs better
with schedules A and C
without row reordering.
In fact, with schedule C
the \ubalg algorithm
attains a rather nice result
by having a accuracy of $1.00$ with an
average cost of $0.62$.

\subsection{Sensitivity to the parameter $\alpha$}
\label{sect:alphatest}

We continue by
studying the sensitivity of the \probabalg algorithm
to the value of $\alpha$.
A method for selecting a good value of $\alpha$
was proposed
in Section~\ref{section:alpha}.
We compare this 
value, denoted $\alpha^*$,
with the values $2\alpha^*$ and $\frac 12 \alpha^*$.
In addition to the actual values of cost and accuracy,
we also show two other quantities,
denoted $g_\downarrow$ and $g_\uparrow$.
These indicate the ratio of the relative
change in accuracy to the relative change in the cost
when $\alpha$ is halved or doubled, respectively.
We let
$g_\downarrow = \frac{A_{\alpha^*/2}/A_{\alpha^*}}{C_{\alpha^*/2}/C_{\alpha^*}}$, and
$g_\uparrow = \frac{A_{2\alpha^*}/A_{\alpha^*}}{C_{2\alpha^*}/C_{\alpha^*}}$.
When $g_\downarrow < 1$
the relative increase in accuracy
is less than the relative increase in costs.
Respectively, when $g_\uparrow > 1$
the relative decrease in accuracy
is larger than the relative decrease in costs.
On the other hand,
when either $g_\downarrow > 1$ or $g_\uparrow < 1$
it would be more efficient to
use $\alpha^*/2$ or $2\alpha^*$ instead of $\alpha^*$.
\begin{table}
\centering
\caption{Cost (top) and accuracy (middle) in random data with the PR algorithm for different $\alpha$.}
\label{tbl:alpha_rnd_results}
\scriptsize
\begin{tabular}{l|r@{\hspace{2mm}}r@{\hspace{2mm}}r@{\hspace{2mm}}r@{\hspace{2mm}}r}
           & A           & B           & C           & D           & learned \\
\hline
$\alpha^*/2$ & 0.46 (0.19) & 0.26 (0.10) & 0.37 (0.13) & 0.27 (0.11) & 0.27 (0.11)\\
$\alpha^*$   & 0.40 (0.18) & 0.24 (0.09) & 0.32 (0.12) & 0.24 (0.10) & 0.23 (0.10)\\
$2 \alpha^*$ & 0.34 (0.16) & 0.22 (0.08) & 0.25 (0.10) & 0.20 (0.08) & 0.19 (0.08)\\
\hline
\hline
$\alpha^*/2$ & 0.87 (0.10) & 0.88 (0.11) & 0.90 (0.09) & 0.91 (0.09) & 0.88 (0.11)\\
$\alpha^*$   & 0.82 (0.14) & 0.84 (0.13) & 0.85 (0.11) & 0.86 (0.11) & 0.84 (0.13)\\
$2 \alpha^*$ & 0.74 (0.16) & 0.79 (0.15) & 0.75 (0.14) & 0.80 (0.14) & 0.75 (0.17)\\
\hline
\hline
$g_\downarrow$ & 0.92        & 0.97        & 0.91        & 0.94        & 0.89\\
$g_\uparrow$  & 1.06        & 1.02        & 1.13        & 1.11        & 1.08
\end{tabular}

\end{table}

\begin{table}
\centering
\caption{Cost (top) and accuracy (middle) in Wikipedia with the PR algorithm for different $\alpha$.}
\label{tbl:alpha_wiki_results}
\scriptsize
\begin{tabular}{l|r@{\hspace{2mm}}r@{\hspace{2mm}}r@{\hspace{2mm}}r@{\hspace{2mm}}r}
             & A           & B           & C           & D          & learned\\
\hline
$\alpha^*/2$ & 0.46 (0.30) & 0.47 (0.10) & 0.73 (0.24) & 0.55 (0.23) & 0.73 (0.28)\\
$\alpha^*$   & 0.30 (0.24) & 0.42 (0.08) & 0.60 (0.27) & 0.48 (0.22) & 0.63 (0.30)\\
$2 \alpha^*$ & 0.19 (0.01) & 0.37 (0.05) & 0.28 (0.29) & 0.35 (0.22) & 0.42 (0.29)\\
\hline
\hline
$\alpha^*/2$ & 0.81 (0.21) & 0.89 (0.17) & 0.88 (0.21) & 0.92 (0.18) & 0.92 (0.16)\\
$\alpha^*$   & 0.62 (0.27) & 0.82 (0.21) & 0.79 (0.24) & 0.88 (0.22) & 0.84 (0.21)\\
$2 \alpha^*$ & 0.37 (0.25) & 0.69 (0.23) & 0.46 (0.33) & 0.76 (0.27) & 0.64 (0.28)\\
\hline
\hline
$g_\downarrow$ & 0.85       & 0.97        & 0.92        & 0.91        & 0.95\\
$g_\uparrow$   & 0.94       & 0.96        & 1.25        & 1.18        & 1.14\\
\end{tabular}

\end{table}

Results for random data are shown in Table~\ref{tbl:alpha_rnd_results}.
Costs are shown in the top part of the table,
while accuracy is shown in the middle part.
As expected,
halving (doubling) the value of $\alpha^*$
increases (decreases) both cost and accuracy.
However, as indicated by $g_\downarrow$ and $g_\uparrow$,
the increase (decrease) in accuracy
is never large (small) enough
to warrant the corresponding increase (decrease) in the cost.
Table~\ref{tbl:alpha_wiki_results} shows the results for Wikipedia.
The behavior is the same as with random data,
with the exception that now
$g_\uparrow$ is below $1$ for schedules A and B,
indicating that in this case the relative
gain in decreased cost is larger than the
relative loss in decreased accuracy.
Indeed, using schedule B
(rank attributes in decreasing order of the absolute value of $\w_j$)
with $\alpha$ is set to $2 \alpha^*$
we obtain
an average cost of $0.37$ with an average accuracy of $0.69$,
which can still be considered a
reasonable performance.

\subsection{Correlated weights and costs}

This experiment is only ran using random data.
We want to study how the relationship of $\w$
and $C(A_j)$ affects the performance of the algorithms.
We are interested in the case where
the most important attributes according to $\w$,
i.e.~those with the highest absolute values,
also have the highest costs.
In this case the baseline schedules B and C
(see Section~\ref{section:schedule})
disagree as much as possible.
The experiment is ran in the same way as
the one in Section~\ref{sect:exp_schedule},
with the exception that we let $C(A_j) = \w_j$.
The row reordering heuristic is being used.
\begin{table}[t]
\centering
\caption{Costs (top) and accuracies (bottom) for the algorithms when $C(A_j) = \w_j$ for $k = 10$.}
\label{tbl:weight_cost_results}
\scriptsize
\begin{tabular}{l|r@{\hspace{2mm}}r@{\hspace{2mm}}r@{\hspace{2mm}}r@{\hspace{2mm}}r}
           &   A         &  B          & C           & D           & learned \\
\hline
\ubalg     & 0.86 (0.04) & 0.81 (0.03) & 0.91 (0.04) & 0.85 (0.04) & - \\
\mproalg   & 0.72 (0.07) & 0.78 (0.03) & 0.64 (0.05) & 0.72 (0.07) & - \\
\probabalg & 0.47 (0.10) & 0.37 (0.05) & 0.62 (0.12) & 0.46 (0.11) & 0.44 (0.10)\\
\hline
\hline
\ubalg     & 1.00 (0.00) & 1.00 (0.00) & 0.99 (0.03) & 0.99 (0.01) & - \\
\mproalg   & 0.69 (0.23) & 0.92 (0.10) & 0.36 (0.15) & 0.67 (0.23) & - \\
\probabalg & 0.82 (0.12) & 0.83 (0.13) & 0.80 (0.15) & 0.82 (0.14) & 0.78 (0.17)
\end{tabular}

\end{table}

Results are shown in Table~\ref{tbl:weight_cost_results}.
Costs are given in the upper part of the table,
while accuracies are shown in the lower part.
Clearly the \probabalg algorithm still outperforms both
baselines with every schedule.
However, when the numbers are compared with those in
tables \ref{tbl:schedule_randomdata_cost} and
\ref{tbl:schedule_randomdata_accuracy},
we observe a noticeable decrease in performance.
The average costs of C, D, and the learned schedule
are twice as high when
the most important features also have the highest costs.
But even now the average cost of a query is less than 50 percent
of the full cost with the \probabalg algorithm.

\subsection{Effect of k}

The performance of the algorithms may
also depend on the size of the top-$k$ set.
For smaller $k$ we expect the pruning
to be more efficient,
as the threshold $\delta$ is larger.
In addition to $k=10$ that was used in the previous experiments,
we also run the algorithms with $k=5$ and $k=20$
to see how this affects the results.
In this test we only consider the
weight-cost ratio heuristic (D) for the schedule.

\begin{table}
\centering
\caption{Costs (top) and accuracies (bottom) with random input, schedule D and different $k$.}
\label{tbl:ktest_rnd}
\begin{tabular}{l|rrr}
           & k = 5       & k = 10      & k = 20 \\
\hline
\ubalg     & 0.85 (0.08) & 0.88 (0.07) & 0.93 (0.05)\\
\mproalg   & 0.64 (0.17) & 0.66 (0.13) & 0.68 (0.12)\\
\probabalg & 0.19 (0.09) & 0.23 (0.10) & 0.29 (0.10)\\
\hline
\hline
\ubalg     & 1.00 (0.00) & 1.00 (0.00) & 1.00 (0.00)\\
\mproalg   & 0.87 (0.20) & 0.88 (0.11) & 0.89 (0.12)\\
\probabalg & 0.87 (0.14) & 0.85 (0.16) & 0.86 (0.09)
\end{tabular}

\end{table}
\begin{table}[t]
\centering
\caption{Costs (top) and accuracies (bottom) with Wikipedia, schedule D and different $k$.}
\label{tbl:ktest_wiki}
\begin{tabular}{l|rrr}
           & k = 5       & k = 10      & k = 20 \\
\hline
\ubalg     & 0.83 (0.00) & 0.91 (0.05) & 0.84 (0.00)\\
\mproalg   & 0.82 (0.00) & 0.82 (0.00) & 0.83 (0.00)\\
\probabalg & 0.41 (0.22) & 0.43 (0.22) & 0.43 (0.22)\\
\hline
\hline
\ubalg     & 0.98 (0.05) & 1.00 (0.00) & 0.99 (0.01)\\
\mproalg   & 0.99 (0.03) & 0.99 (0.01) & 1.00 (0.00)\\
\probabalg & 0.83 (0.28) & 0.83 (0.23) & 0.79 (0.25)
\end{tabular}

\end{table}

Table~\ref{tbl:ktest_rnd} shows results for
random data.
Clearly the cost increases as $k$ increases.
Especially for the \probabalg algorithm
the effect is considerable.
However, accuracy is not really affected
for any of the algorithms.
Results for Wikipedia are shown in Table~\ref{tbl:ktest_wiki}.
Here we do not see any significant effect
on either the cost or accuracy.

\section{Conclusion and future work}

We have discussed an algorithm for
approximate top-$k$ search
in a setting where
the input relation is initially hidden,
and its elements can be accessed only
by paying a (usually computational) cost.
The score of a row is defined
as its inner product with a scoring vector.
The basic task is to find an
approximate top-$k$ set
while keeping the total cost of the query low.
Although we consider linear scoring functions in this paper,
the proposed approach should yield itself also to
other types of of aggregation functions.

Since the contents of the relation are unknown
before any queries are issued,
indexing its contents is not possible.
This is a key property of our setting
that differentiates it from most of existing literature
on top-$k$ as well as $k$-NN search.
Instead we have access to 
training data from the same distribution as the hidden relation.
The algorithm we propose is based on the use of this training data.
Given the partial score of an item,
the algorithm estimates the probability
that the full score of the item will be high enough for the
item to enter the current top-$k$ set.
The estimator for this probability is learned from training data.
The algorithm has two parameters.
We also propose methods for
learning good values for these from training data.
The experiments indicate that our proposed algorithm
outperforms the baseline
in terms of the cost
by a considerable margin.
While the MPro \cite{mpro} algorithm
attains a very high accuracy it does this at a high cost.

The work presented in this paper
is mostly related to databases and
approximate nearest neighbor search.
However, we also want to point out
some connections to classification problems,
and especially feature selection.
Our approach can be seen as a form of
dynamic feature selection
for top-$k$ problems
with the aim to reduce the overall cost of the query.
Similarly we can consider
cost-sensitive classification
(see e.g.~\cite{davis06ecml}),
where the task is to classify
a given set of items
while keeping the total cost as low as possible.
Based on a subset of the available features
the classifier makes an initial prediction,
and if this prediction is
not certain enough,
we read the value of a yet unknown feature
and update the prediction accordingly.
Decision trees already implement this principle in a way,
but it might be interesting to
extend it to other classification algorithms,
such as SVMs.

Conversely,
a potentially interesting approach to extending
the work of this paper
is to replace the linear schedule
with something more complex,
such as a decision tree.
In this case
the next attribute to be read
would depend on the value(s) of the previous attribute(s).
The results of \cite{charikar02,kaplan05lac,cicalese05,greiner06}
might provide a fruitful starting point
for studying the theoretical properties of the problem.
Further studies include the
use of more complex models than linear regression
for estimating $\mu$ and $\sigma$.
Also, using other distributions than a Gaussian
for the full score given a prefix score
may be of interest.


\bibliographystyle{abbrv} \bibliography{topk09}

\begin{thebibliography}{10}

\bibitem{ailon06}
N.~Ailon and B.~Chazelle.
\newblock Approximate nearest neighbors and the fast johnson-lindenstrauss
  transform.
\newblock In {\em Proceedings of the 38th Annual ACM Symposium on Theory of
  Computing (STOC 2006)}, pages 557--563, 2006.

\bibitem{charikar02}
M.~Charikar, R.~Fagin, V.~Guruswami, J.~Kleinberg, P.~Raghavan, and A.~Sahai.
\newblock Query strategies for priced information.
\newblock {\em Journal of Computer and System Sciences}, 64(4), 2002.

\bibitem{cicalese05}
F.~Cicalese and E.~S. Laber.
\newblock A new strategy for querying priced information.
\newblock In {\em Proceedings of the thirty-seventh annual ACM symposium on
  Theory of computing}, pages 674--683, 2005.

\bibitem{das06answering}
G.~Das, D.~Gunopulos, N.~Koudas, and D.~Tsirogiannis.
\newblock Answering top-$k$ queries using views.
\newblock In {\em Proceedings of the 32nd international conference on Very
  large data bases}, pages 451--462, 2006.

\bibitem{davis06ecml}
J.~V. Davis, J.~Ha, C.~R. Rossbach, H.~E. Ramadan, and E.~Witchel.
\newblock Cost-sensitive decision tree learning for forensic classification.
\newblock In {\em Proceedings of the 17th European Conference on Machine
  Learning}, pages 622--629, 2006.

\bibitem{fagin2001}
R.~Fagin, A.~Lotem, and M.~Naor.
\newblock Optimal aggregation algorithms for middleware.
\newblock {\em Journal of Computer and System Sciences}, 66(4):614--656, 2003.

\bibitem{feller50}
W.~Feller.
\newblock {\em An Introduction to Probability Theory and Its Applications},
  volume~I.
\newblock John Wiley \& Sons, 3rd edition, 1968.

\bibitem{gionis99lsh}
A.~Gionis, P.~Indyk, and R.~Motwani.
\newblock Similarity search in high dimensions via hashing.
\newblock In {\em Proceedings of the 25th International Conference on Very
  Large Data Bases}, pages 518--529, 1999.

\bibitem{goel08}
S.~Goel, J.~Langford, and A.~Strehl.
\newblock Predictive indexing for fast search.
\newblock In {\em Proceedings of the Twenty-Second Annual Conference on Neural
  Information Processing Systems (NIPS 2008)}, pages 505--512, 2008.

\bibitem{greiner06}
R.~Greiner, R.~Hayward, M.~Jankowska, and M.~Molloy.
\newblock Finding optimal satisficing strategies for and-or trees.
\newblock {\em Artificial Intelligence}, 170(1):19--58, 2006.

\bibitem{gunzer00optimizing}
U.~G{\"u}ntzer.
\newblock Optimizing multi-feature queries for image databases.
\newblock In {\em Proceedings of the 26th International Conference on Very
  Large Data Bases}, pages 419--428, 2000.

\bibitem{hellerstein93}
J.~M. Hellerstein and M.~Stonebraker.
\newblock Predicate migration: Optimizing queries with expensive predicates.
\newblock In {\em Proceedings of the 1993 ACM SIGMOD International Conference
  on Management of Data}, pages 267--276, 1993.

\bibitem{hristides01sigmod}
V.~Hristidis, N.~Koudas, and Y.~Papakonstantinou.
\newblock Prefer: A system for the efficient execution of multi-parametric
  ranked queries.
\newblock In {\em ACM SIGMOD Record}, pages 259--270, 2001.

\bibitem{mpro}
S.~Hwang and K.~C.-C. Chang.
\newblock Probe minimization by schedule optimization: Supporting top-k queries
  with expensive predicates.
\newblock {\em IEEE Transactions on Knowledge and Data Engineering},
  19(5):646--662, 2007.

\bibitem{topksurvey}
I.~F. Ilyas, G.~Beskales, and M.~A. Soliman.
\newblock A survey of top-$k$ query processing techniques in relational
  database systems.
\newblock {\em ACM Computing Surveys}, 40(4), 2008.

\bibitem{indyk04survey}
P.~Indyk.
\newblock {\em Handbook of Discrete and Computational Geometry}, chapter
  Nearest neighbors in high-dimensional spaces.
\newblock CRC Press, 2004.

\bibitem{kaplan05lac}
H.~Kaplan, E.~Kushilevitz, and Y.~Mansour.
\newblock Learning with attribute costs.
\newblock In {\em Proceedings of the thirty-seventh annual ACM symposium on
  Theory of computing (STOC 2005}, pages 356--365, 2005.

\bibitem{kemper94}
A.~Kemper, G.~Moerkotte, K.~Peithner, and M.~Steinbrunn.
\newblock Optimizing disjunctive queries with expensive predicates.
\newblock In {\em Proceedings of the 1994 ACM SIGMOD international conference
  on Management of data}, pages 336--347, 1994.

\bibitem{klei99knn}
J.~Kleinberg.
\newblock Two algorithms for nearest-neighbor search in high dimensions.
\newblock In {\em Proceedings of the twenty-ninth annual ACM symposium on
  Theory of computing (STOC 1999)}, pages 599--608, 1999.

\bibitem{kumar09topk}
R.~Kumar, K.~Punera, T.~Suel, and S.~Vassilvitskii.
\newblock Top-$k$ aggregation using intersections of ranked inputs.
\newblock In {\em Proceedings of the Second ACM International Conference on Web
  Search and Data Mining}, pages 222--231, 2009.

\bibitem{marian04tods}
A.~Marian, N.~Bruno, and L.~Gravano.
\newblock Evaluating top-k queries over web-accessible databases.
\newblock {\em ACM Transactions on Database Systems (TODS)}, 29(2):319--362,
  2004.

\bibitem{natsev01vldb}
A.~Natsev, Y.-C. Chang, J.~R. Smith, C.-S. Li, and J.~S. Vitter.
\newblock Supporting incremental join queries on ranked inputs.
\newblock In {\em Proceedings of the 27th International Conference on Very
  Large Data Bases}, pages 281--290, 2001.

\bibitem{nepal99query}
S.~Nepal and M.~V. Ramakrishna.
\newblock Query processing issues in image (multimedia) databases.
\newblock In {\em Proceedings of the 15th International Conference on Data
  Engineering}, pages 22--29, 1999.

\bibitem{russellnorvig}
S.~Russell and P.~Norvig.
\newblock {\em Artifical Intelligence: A Modern Approach}.
\newblock Prentice Hall, 2nd edition, 2003.

\bibitem{singitham04}
P.~Singitham, M.~Mahabhashyam, and P.~Raghavan.
\newblock Efficiency-quality tradeoffs for vector score aggregation.
\newblock In {\em Proceedings of the Thirtieth international conference on Very
  large data bases (VLDB 2004)}, pages 624--635, 2004.

\bibitem{theobald04topk}
M.~Theobald, G.~Weikum, and R.~Schenkel.
\newblock Top-$k$ query evaluation with probabilistic guarantees.
\newblock In {\em Proceedings of the 30th VLDB Conference}, pages 648--659,
  2004.

\bibitem{ukkonen08cikm}
A.~Ukkonen, C.~Castillo, D.~Donato, and A.~Gionis.
\newblock Searching the wikipedia with contextual information.
\newblock In {\em Proceedings of the 17th ACM Conference on Information and
  Knowledge Management, CIKM 2008}, pages 1351--1352, 2008.

\bibitem{kernelbook}
M.~P. Wand and M.~C. Jones.
\newblock {\em Kernel Smoothing}.
\newblock Chapman \& Hall, CRC Press, 1994.

\end{thebibliography}

\end{document}